\documentclass[amsmath,amssymb,14pt]{article}
\usepackage{graphicx, epsfig}

\usepackage{amssymb,amsmath,xcolor}
\usepackage{amsthm}
\usepackage[a4paper, total={6.5in, 10in}]{geometry}

\usepackage{tikz}
\usetikzlibrary{patterns, quotes, positioning, decorations.markings}

\usetikzlibrary{arrows.meta,
                decorations.markings,
                patterns.meta}

\numberwithin{equation}{section}

\usepackage{subcaption}
\usepackage{caption}
\usepackage[export]{adjustbox}

\usepackage[nottoc]{tocbibind}

\usepackage{dcolumn}
\usepackage{bm}

\newtheorem{mytheo}{Theorem}[section]

\newtheorem{myproposition}[mytheo]{Proposition}

\newtheorem{myremark}[mytheo]{Remark}

\newtheorem{mytheo*}{Theorem}
\newtheorem{mylemma*}{Lemma}
\newtheorem{myproposition*}{Proposition}
\newtheorem{mydef*}{Definition}[section]
\newtheorem{myremark*}{Remark}
\newtheorem*{myproof*}{Sketch of the Proof}
\newtheorem*{mynotation*}{Notation}
\newtheorem{mycorollary*}{Corollary}
\newtheorem*{theorem*}{Theorem}

\usepackage{xcolor}

\DeclareMathOperator{\Imaginary}{Im}
\DeclareMathOperator{\Real}{Re}

\usepackage{mathtools}

\usepackage{authblk}

\begin{document}
\title{Completeness Relation in Renormalized Quantum Systems}

\author[1]{Fatih Erman}
\author[2, 3]{O. Teoman Turgut}
\affil[1]{Department of Mathematics, \.{I}zmir Institute of Technology, Urla, 35430, \.{I}zmir, Turkey}
\affil[2]{Department of Physics, Bo\u{g}azi\c{c}i University, Bebek, 34342, \.{I}stanbul, Turkey}
\affil[3]{Department of Physics, Carnegie Mellon University, Pittsburgh, PA, United States}

\affil[1]{fatih.erman@gmail.com}
\affil[2]{turgutte@bogazici.edu.tr}

\maketitle

\begin{abstract} In this work, we show that the completeness relation for the eigenvectors, which is an essential assumption of quantum mechanics, remains true if the Hamiltonian, having a discrete spectrum, is modified by a delta potential (to be made precise by a renormalization scheme) supported at a point in two and three-dimensional compact manifolds or Euclidean spaces. The formulation can be easily extended
to $N$ center case, and the case where delta interaction is supported on curves in the plane or space. We finally give an interesting application for sudden perturbation of the support of the delta potential. \end{abstract}

\section{Introduction} \label{introduction}

In quantum mechanics, the energy eigenfunctions — corresponding to both discrete and continuous spectra — constitute a generalized orthonormal basis for the Hilbert space $\mathcal{H}$. This allows any arbitrary state (wave function) to be expanded in terms of these generalized eigenfunctions, a fundamental property known as the completeness relation (or Parseval’s identity for eigenfunctions) \cite{Arnobook, GalindoPascual1, Berezanskii}. There are only a few standard explicit examples in which the completeness relation has been verified. One of the reasons for this is the lack of exactly solvable potentials in quantum mechanics and the integrals or sums involving eigenfunctions are quite hard to evaluate analytically. The momentum operator and the Hamiltonian for a single particle in a box are the most well-known textbook examples \cite{Griffiths, Gasiorowicz}. The completeness relation for systems having both bound states and continuum states, such as the Dirac delta potential in one dimension \cite{Brownstein, Patil, Dalabeeh}, the Coulomb potential in three dimensions \cite{Mukunda}, and the reflectionless potential \cite{reflectionless} have also been demonstrated by appropriately normalizing the eigenfunctions. The purpose of this paper is to show that the completeness relation still holds even for $a$ rather singular system, involving delta function potentials, where the renormalization is required. For this, we consider an Hamiltonian having only a discrete spectrum and assume (justifiably for a self-adjoint Hamiltonian) that the completeness relation holds. Then we prove that the completeness relation is still true even if we modify this Hamiltonian by a delta potential (point interactions in two and three dimensions in an Euclidean space, as well as point interactions in two and three-dimensional compact manifolds), where a renormalization is required to render the Hamiltonian well-defined.

The resolvent of the modified Hamiltonian by singular delta potentials supported by a point $a$ in two or three dimensions has been studied extensively in the literature and given by the Krein's formula \cite{Albeverio2012solvable, AlbeverioKurasov}
\begin{eqnarray}
    R(E)=R_0(E) + (\Phi(E))^{-1} \langle \overline{G_0(\cdot, a|E)}, \cdot \rangle G_0(\cdot, a|E) \;, \label{resolventexplicitformula}
\end{eqnarray}
where $R_0(E)=(H_0-E)^{-1}$ is the resolvent of the Hamiltonian $H_0$ at $E \notin \mathbb{R}$, $G_0(x,y|E)$ is the integral kernel of the resolvent $R_0(E)$ or Green's function, and $\Phi$ is some function to be determined for each particular class of singular potential. This function is also denoted by $\Gamma$ in the mathematics literature. The meaning of the second term should be understood as follows: 
\begin{eqnarray}
    (R(E) \psi)(x)=(R_0(E)\psi)(x) + (\Phi(E))^{-1} \langle \overline{G_0(y, a|E)}, \psi(y) \rangle G_0(x, a|E) \;, 
\end{eqnarray}
where $\langle \overline{G_0(y, a|E)}, \psi(y) \rangle = \int G_0(y, a|E) \psi(y) d \mu(y)$. The formula (\ref{resolventexplicitformula}) can be seen more naturally in Dirac's bra-ket notation, 
\begin{eqnarray}
    R(E) = R_0(E) + (\Phi(E))^{-1} R_0(E) |a \rangle \langle a | R_0(E) \;.
\end{eqnarray}
Looking at the resulting wave functions, some of our colleagues express doubts about the explicit verification of the completeness relations, even though it was clear from the fact that the resulting Hamiltonians are self-adjoint in a precise mathematical sense. Even if the result is expected, we think it is a valuable exercise to demonstrate the orthonormality and completeness by an explicit calculation. 
To make the presentation self-contained, we will briefly summarize how the pole structure of the full Green's function $G(x,y|E)=\langle x |R_0(E)|y\rangle$ is rearranged to form new poles and how the poles of $G_0(x,y|E)$, which explicitly appears as an additive factor in $G(x,y|E)$, are removed in general. This has been proved in our previous work \cite{Annals23} for the more general case, when the Hamiltonian has a discrete as well as a continuous spectrum.

The resulting wave functions are typically given by the original Green's functions $G_0$ evaluated at the new energy eigenvalues, so they are actually (mildly) singular at the location of the delta function. These are interesting objects by themselves and could be useful in some practical problems as well, as they are now (explicitly) shown to form a new orthonormal basis. In the present work, we prefer to emphasize the essential ideas while writing out  our proofs and we are not aiming for a fully  rigorous mathematical approach, in this way,  we hope that, the paper becomes accessible to  a wider audience.

\section{Discrete Spectrum Modified by a  $\delta$ Interactions} \label{section2}

To set the stage, we introduce the notation and summarize the main results about how the spectrum of an initial Hamiltonian $H_0$ having a purely discrete spectrum changes under the influence of a (formally defined) delta interaction,  which was discussed in our previous works, particularly in \cite{Annals23}.

We consider the case in which $H_0$ is {\it formally} modified by a single $\delta$ function supported at $x=a$,
\begin{eqnarray}
    H=H_0 - \alpha \delta_a \;, \label{Hamiltonian}
\end{eqnarray}
where $\alpha$ is to be replaced by a renormalized coupling once we actually state the Green's function for this problem. Various methods exist in literature to make sense of the above formal expression of the Hamiltonian $H$. One possible way is to define the $\delta$ interaction as a self-adjoint extension of $H_0$ and they are in general called point interactions or contact interactions. A modern introduction to this subject is given in the recent book by Gallone and Michelangeli \cite{allesandro} and the classic reference elaborating this point of view is the monograph by Albeverio et. al. \cite{Albeverio2012solvable}.

Here and subsequently, as emphasized in the introduction, we assume that the initial Hamiltonian $H_0$ satisfies some conditions:
\begin{itemize}
    \item $H_0$ is self-adjoint on some dense domain $D(H_0) \subset L^2({\mathcal M})$, where ${\mathcal M}$ is two or three dimensional Euclidean space or Riemannian compact manifold without boundary (connected and orientable additionally). Often, it is essential (to put some estimates on the Green's functions) to assume some regularity on the geometry, experience has shown that a lower bound on the Ricci curvature, which controls the volume growth of geodesic balls, satisfies most of the technical requirements. Consequently, we impose the following condition,
\begin{equation}
    Ric_g(\cdot, \cdot)\geq (D-1)\kappa g(\cdot, \cdot) \;. \label{ricci}
\end{equation}
For two dimensional compact manifolds, this does not impose any restriction, as Ricci curvature is exactly given by 
$Ric_g(\cdot, \cdot)={R\over 2} g(\cdot, \cdot)$, where $R$ is the scalar curvature, and $R$ has a minimum (and a maximum) value on a {\it compact} manifold.  For three dimensional manifolds, this puts some restriction on possible geometric structures one admits.
If $\kappa>0$, one has much better control for various bounds on heat kernels (or Green's functions), see the book by Li \cite{peterLi} for an exposition of these ideas.
  
    \item Spectrum of $H_0$ is discrete $\sigma_d(H_0)$ (set of eigenvalues),
    \item The discrete spectrum has no accumulation point, except possibly at infinity.
    \item For stability, we assume $H_0$ has spectrum  bounded below.
\end{itemize}
These conditions on the spectrum put some mild restrictions on the potential $V$ (listed in the classical work of Reed and Simon \cite{ReedSimonv4}) if we assume 
\begin{equation}
    H_0=-{\hbar^2\over 2m} \Delta+V \;,
\end{equation}
on $D=2,3$  dimensional Euclidean space, and they are true when we consider 
\begin{equation}
    H_0=-{\hbar^2\over 2m} \Delta_g \;,
\end{equation}
on a compact Riemannian manifold (again of dimension 2 or 3) with a metric $g_{ij}$, where $\Delta_g$ is the Laplace-Beltrami operator or Laplacian given by
\begin{eqnarray}
    (\Delta_g \psi)(x)=  \frac{1}{\sqrt{\det g}} \sum_{i,j=1}^{D} \frac{\partial}{\partial x^i} \left( \sqrt{\det g} g^{ij} \frac{\partial \psi(x)}{\partial x^j}\right)   \;,
\end{eqnarray}
in some local coordinates, with $g^{ij}$ being the components of inverse of the metric $g$.
Precisely speaking, it is well known 
\cite{Rosenberg, chavel2} that there exists a complete orthonormal system of $C^\infty$
eigenfunctions $\{\phi_n \}_{n=0}^{\infty}$ in $L^2({\mathcal M})$ and
the spectrum $\sigma(H_0)=\{E_n\} = \{0 = E_0
\leq E_1 \leq E_2 \leq \dots\}$, with $E_n$ tending
to infinity as $n \rightarrow \infty$ and each eigenvalue has
finite multiplicity.
Some eigenvalues are repeated according to their multiplicity. The
multiplicity of the first eigenvalue $E_0=0$ is one and 
the corresponding eigenfunction is constant. From now on, {\it we  assume that there is no degeneracy in the spectrum of the Laplacian for simplicity}. The analysis about how the spectrum changes under the modification of $\delta$ potentials in the presence of degeneracy has been given in Appendix D of our previous work \cite{Annals23}.

\begin{myremark}
Note that the complete nondegeneracy assumption of the spectrum is not an exceptional case. If we introduce  a proper distance in the space of all smooth metrics on the manifold, then the set of metrics with completely non-degenerate spectra are actually dense in this metric space. Incidentally, the space of all smooth metrics becomes what is called a Frechet space under this particular choice of the distance function \cite{Urakawa-thebook}. 
\end{myremark}

\begin{myremark}
  There are upper bounds on the eigenvalues of the Laplacian given in terms of the geometric data, and these give some valuable information about the way the spectrum behaves (for example see Corollary 4.15 of \cite{Urakawa-thebook})
\end{myremark}

The integral kernel of the resolvent $R_{0}(E)$ for $H_0$ or simply Green's function is given by
\begin{eqnarray}
\left(R_0(E)\psi\right)(x) = \left((H_0-E)^{-1}\psi\right)(x)= \int_{\mathcal{M}} G_0(x,y|E)\psi(y) d\mu(y)   \;,
\end{eqnarray}
where $d\mu(y)$ is the volume element in  $\mathcal M$ (on a manifold, expressed in local coordinates, it has the usual $\sqrt{\det g}$ factor in it) and it can be expressed by the following expression away from the diagonal $x=y$,
\begin{eqnarray}
    G_0(x,y|E) = \sum_{n=0}^{\infty} \frac{\phi_n(x) \overline{\phi_n(y)}}{E_n-E}  \;, \label{greenfuncexpansion}
\end{eqnarray}
where $\{\phi_n\}$ is the complete set of eigenfunctions of $H_0$. The Green's function $G_0(x,y|E)$ is a square-integrable function of $x$ for almost all values of $y$ and vice versa \cite{ReedSimonv3}.

When the co-dimension (dimension of the space - dimension of the support of the $\delta$ interaction) is greater than one, $\delta$ interaction must be defined by a renormalization procedure. The main reason for this is based on the singular structure of the Green's function for initial Hamiltonians $H_0$ in two and three dimensions. The history of this subject is quite rich and there has been a vast amount of material in the physics literature, see e.g., \cite{Hoppe, Huang, Jackiw, GosdzinskyTarrach, MeadGodines, ManuelTarrach, Cavalcanti, Coutinho1, Patil, Coutinho2}. An eigenfunction expansion, analogous to (\ref{greenfuncexpansion}), also exists for the   Green's function $G(x,y|E)$ of the modified (formal) Hamiltonian $H$ (a two or three dimensional delta potential added to the free case located at the origin) in \cite{Cavalcanti}. It is possible to express this Green's function $G(x,y|E)$ in terms of the Green's functions of the initial Hamiltonian $H_0$. The standard route in the literature is to construct this Green's function and establish that the Hamiltonian defined by this expression is indeed self-adjoint, hence by {\it the spectral theorem}, there is a complete set of eigenfunctions. In this paper, we prove directly by means of the explicit expression of the constructed Green's function that the corresponding Hamiltonian still has a complete set of eigenfunctions. For this we use the completeness property of the eigenfunctions of the initial Hamiltonian $H_0$,  having only a  discrete spectrum, and an interlacing theorem for the poles of the new Green's function, proved in a previous publication \cite{Annals23}. As a result, we thus establish the self-adjointness  of the resulting Hamiltonian in a novel way (Remark 4.3). Moreover, we have an explicit integral operator for the Hamiltonian, which allows one to apply various approximation methods. There is also great pedagogical value in establishing the existence of an orthonormal basis for a given Hamiltonian, as it demonstrates clearly the validity of one of the fundamental postulates of quantum mechanics. 

It is useful to express  Green's function $G_0$ in terms of the heat kernel $K_t(x,y)$ associated with the operator $H_0$ under the above assumptions. It is given by 
\begin{eqnarray}
    G_0(x,y|E)= \int_{0}^{\infty} K_t(x,y) e^{t E} dt \;, \label{intrepgreen}
\end{eqnarray}
where $\Real(E)<0$ and $H_0 K_t(x,y)= \frac{\partial}{\partial t}K_t(x,y)$ (and can be defined for other values of $E$ in the complex $E$ plane through analytical continuation). We note  that the first term in the short time asymptotic expansion of the diagonal heat kernel for any self-adjoint elliptic second order differential operator \cite{Gilkey} in $D$ dimensions, is given by
\begin{eqnarray}
 K_t(x,x) \sim  t^{-D/2} \;. \end{eqnarray}
This leads to the divergence around $t=0$ in the diagonal part of the Green's function $G_0(x,x|E)$:  
\begin{eqnarray}
    \int_{0}^{\infty} \frac{e^{-t|E|}}{t^{D/2}} \; dt \;,
\end{eqnarray}
for $D=2,3$. In order to make sense of such singular interactions, one must first regularize the Hamiltonian by introducing a cut-off $\epsilon>0$. A natural way, in particular for compact manifolds, is to replace the $\delta$ function by the heat kernel $K_{\epsilon/2}(x,a)$, which converges to $\delta(x-a)$ as $\epsilon \to 0$ (in the distributional sense). It turns out that the regularized Green's function is given by
\begin{eqnarray}
 G^{\epsilon}(x,y|E)= G^{\epsilon}_{0}(x,y|E) + \frac{G^{\epsilon}_{0}(x,a|E) G^{\epsilon}_{0}(a,y|E)}{\frac{1}{\alpha}-G^{\epsilon}_{0}(a,a|E)} \;,    
\end{eqnarray}
where $G^{\epsilon}_{0}(x,y|E) = \int_{\epsilon}^{\infty} K_t(x,y)e^{t E} dt$ with $\Real(E)<0$. Then, we make the coupling constant $\alpha$ dependent on the cut-off $\epsilon$ in such a way that the regularized Green's function has a non-trivial limit as we remove the cut-off. A natural choice for absorbing the divergent part in a redefinition of the coupling constant is given by  
\begin{eqnarray}
 \frac{1}{\alpha(\epsilon)} = \frac{1}{\alpha_R(M)} + \int_{\epsilon}^{\infty} K_t(a,a) e^{t M} d t \;,
\end{eqnarray}
where $M$ is the renormalization scale and could be
eliminated in favor of a physical parameter by imposing a
renormalization condition. Taking the formal limit as $\epsilon \to 0$, we obtain the Krein's type of formula for the integral kernel of the resolvent or Green's function 
\begin{eqnarray}
G(x,y|E)= G_0(x,y|E) + \frac{G_0(x,a|E) G_0(a,y|E)}{\Phi(E)} \;, \label{fullGreenfunction}
\end{eqnarray}
where $ \Phi(E) = \frac{1}{\alpha_R(M)} + \int_{0}^{\infty} K_t(a,a) \left( e^{t M}- e^{t E} \right) \; d t$. Since the bound state energy of the system can be found from the poles of the Green's function, or equivalently zeroes of the function $\Phi$, there must be a relation among $M$, $\alpha_R(M)$, and the bound state
energy of the particle  (due to the presence of $\delta$ potential), say $-\mu^2$. Note that $\alpha_R$ varies with respect to $M$ in a precise way to keep the physics (e.g., bound state energy) independent of this arbitrary choice \cite{AltunkaynakErmanTurgut, ErmanTurgut}. We set the renormalization scale at $M=-\mu^2$ (thinking of a bound state below $E_0$) for simplicity. Then, 
\begin{eqnarray}
   \Phi(E) & = & \frac{1}{\alpha_R} + \int_{0}^{\infty} K_t(a,a) \left( e^{-t\mu^2}- e^{t E} \right) \; d t \nonumber \\ & = & 
   \frac{1}{\alpha_R} + \sum_{n=0}^{\infty} \left( \frac{|\phi_n(a)|^2}{(E_n+\mu^2)}-\frac{|\phi_n(a)|^2}{(E_n-E)}  \right) \nonumber 
   \\ & = &  \frac{1}{\alpha_R}- \sum_{n=0}^{\infty} \frac{|\phi_n(a)|^2(E+\mu^2)}{(E_n-E)(E_n+\mu^2)}\;. \label{Phisum}
\end{eqnarray}
Here we employ the eigenfunction expansion of the heat kernel $K_t(x,y)=\sum_n \overline{\phi_n(x)} \phi_n(y) e^{-t E_n}$ of the Laplacian. The (uniform) convergence of this sum can be shown by using the upper bounds of the heat kernel and this technical part has been given in the Appendix A of our previous work \cite{Annals23}.

Note that we could have chosen a sharp cut-off as well, as is often done in physics literature, for the above calculations. The momentum (in this case energy eigenvalue of the Laplacian) is limited by a finite large number $\Lambda$ to render infinite sums to finite expressions. We then  employ our subtraction to finally take a limit $\Lambda\to \infty$ to remove this arbitrary cut-off in the physical result. It has been shown in \cite{DuttaRoy} that the connection between observable quantities for such point delta interactions in two and three dimensions do not depend on the renormalization scheme that is used.

Moreover, we have shown in \cite{existence} that there exists a unique densely defined closed operator, say $H$, associated with the resolvent whose integral kernel is given by (\ref{fullGreenfunction}).

Since the truncation of the above sum (\ref{Phisum}) has no zeros on the upper and lower complex $E$ plane, the uniform convergence of this sum on compact subsets of the complex plane, in conjunction with the Hurwitz theorem \cite{Conway}  implies that all the zeros of $\Phi$ are located on the real $E$ axis. Then, the spectrum of the full Hamiltonian (\ref{Hamiltonian}) is given by the following proposition, which is a particular case of our previous result \cite{Annals23}:
\begin{myproposition} \label{prop21}Let $\phi_k(x)$ be the  eigenfunction of $H_0$ associated with the eigenvalue  $E_k$.
Then, the  (new) energy eigenvalue $E_{k}^{*}$ of $H$, is found from the unique solution of  the equation
\begin{eqnarray}
\Phi(E)=\frac{1}{\alpha_R}- \sum_{n=0}^{\infty} \frac{|\phi_n(a)|^2(E+\mu^2)}{(E_n-E)(E_n+\mu^2)}   = 0 \label{boundstatenergyrenorm}\;,
\end{eqnarray}
which lies in between $E_{k-1}$ and $E_k$, if $\phi_k(a) \neq 0$ for this particular $k$. If for this particular choice of $k$, we have $\phi_k(a)=0$, the corresponding energy eigenvalue does  not change, i.e., $E_{k}^{*}=E_k$.  
For the ground state ($k=0$), we always have  $E_{0}^{*} < E_0$.
\end{myproposition}

\begin{proof}
We first split the term in the eigenfunction expansion of the Green's functions $G_0$ and the function $\Phi$ in (\ref{fullGreenfunction}) associated with the isolated simple eigenvalue $E_k$ of $H_0$:
\begin{eqnarray} & & \hskip-1cm
    G(x,y|E)= \sum_{n \neq k} \frac{\phi_n(x) \overline{\phi_n(y)}}{E_n-E} + \frac{\phi_k(x) \overline{\phi_k(y)}}{E_k-E} + \frac{ \left(\sum_{n \neq k} \frac{\phi_n(x) \overline{\phi_n(a)}}{E_n-E} \right) \left( \sum_{n \neq k} \frac{\phi_n(a) \overline{\phi_n(y)}}{E_n-E}\right)}{\frac{1}{\alpha_R}-\sum_{n \neq k} \frac{|\phi_n(a)|^2(E+\mu^2)}{(E_n-E)(E_n+\mu^2)} - \frac{|\phi_k(a)|^2(E+\mu^2)}{(E_k-E)(E_k+\mu^2)}}  \nonumber \\ & & \hspace{1cm}+  \frac{\left(\sum_{n \neq k} \frac{\phi_n(x) \overline{\phi_n(a)}}{E_n-E} \right) \left(\frac{\phi_k(a) \overline{\phi_k(y)}}{E_k-E}\right)}{\frac{1}{\alpha_R}-\sum_{n \neq k} \frac{|\phi_n(a)|^2(E+\mu^2)}{(E_n-E)(E_n+\mu^2)} - \frac{|\phi_k(a)|^2(E+\mu^2)}{(E_k-E)(E_k+\mu^2)}} 
    \nonumber \\ & & \hspace{2cm} + \frac{ \left(\frac{\phi_k(x) \overline{\phi_k(a)}}{E_k-E} \right) \left( \sum_{n \neq k} \frac{\phi_n(a) \overline{\phi_n(y)}}{E_n-E}\right)}{\frac{1}{\alpha_R}-\sum_{n \neq k} \frac{|\phi_n(a)|^2(E+\mu^2)}{(E_n-E)(E_n+\mu^2)} - \frac{|\phi_k(a)|^2(E+\mu^2)}{(E_k-E)(E_k+\mu^2)}} \nonumber \\ & & \hspace{3cm} + \frac{ \left(\frac{\phi_k(x) \overline{\phi_k(a)}}{E_k-E} \right) \left( \frac{\phi_k(a) \overline{\phi_k(y)}}{E_k-E}\right)}{\frac{1}{\alpha_R}-\sum_{n \neq k} \frac{|\phi_n(a)|^2(E+\mu^2)}{(E_n-E)(E_n+\mu^2)} - \frac{|\phi_k(a)|^2(E+\mu^2)}{(E_k-E)(E_k+\mu^2)}} \;.
\end{eqnarray}
If we combine the second and the last term in the above expression, we obtain
\begin{eqnarray} & & \hskip-1cm
    G(x,y|E)= \frac{\phi_k(x) \overline{\phi_k(y)}}{E_k-E}  \left(1- \left(1- \frac{(E_k-E)}{|\phi_k(a)|^2} \left(\frac{1}{\alpha_R}-\sum_{n \neq k} \frac{|\phi_n(a)|^2 (E+\mu^2)}{(E_n-E)(E_n +\mu^2)} + \frac{|\phi_k(a)|^2}{E_k+\mu^2}\right)\right)^{-1}\right) \nonumber \\ & & \hspace{1cm} + \sum_{n \neq k} \frac{\phi_n(x) \overline{\phi_n(y)}}{E_n-E} + (E_k-E) \frac{ \left(\sum_{n \neq k} \frac{\phi_n(x) \overline{\phi_n(a)}}{E_n-E} \right) \left( \sum_{n \neq k} \frac{\phi_n(a) \overline{\phi_n(y)}}{E_n-E}\right)}{(E_k-E) \left(\frac{1}{\alpha_R}-\sum_{n \neq k} \frac{|\phi_n(a)|^2(E+\mu^2)}{(E_n-E)(E_n+\mu^2)}\right) - \frac{|\phi_k(a)|^2(E+\mu^2)}{(E_k+\mu^2)}}  \nonumber \\ & & \hspace{2cm} +  \frac{\left(\sum_{n \neq k} \frac{\phi_n(x) \overline{\phi_n(a)}}{E_n-E} \right) \left(\phi_k(a) \overline{\phi_k(y)}\right)}{(E_k-E)\left(\frac{1}{\alpha_R}-\sum_{n \neq k} \frac{|\phi_n(a)|^2(E+\mu^2)}{(E_n-E)(E_n+\mu^2)}\right) - \frac{|\phi_k(a)|^2(E+\mu^2)}{(E_k-E)(E_k+\mu^2)}} 
    \nonumber \\ & & \hspace{3cm} + \frac{ \left(\phi_k(x) \overline{\phi_k(a)} \right) \left( \sum_{n \neq k} \frac{\phi_n(a) \overline{\phi_n(y)}}{E_n-E}\right)}{(E_k-E) \left(\frac{1}{\alpha_R}-\sum_{n \neq k} \frac{|\phi_n(a)|^2(E+\mu^2)}{(E_n-E)(E_n+\mu^2)}\right) - \frac{|\phi_k(a)|^2(E+\mu^2)}{(E_k+\mu^2)}} \;.
\end{eqnarray}
Except for the first term, it is easy to see that all terms are regular near $E=E_k$. For the first term, if we choose $E$ sufficiently close to $E_k$, i.e., if $\frac{|E_k-E|}{|\phi_k(a)|^2} \left|\frac{1}{\alpha_R}-\sum_{n \neq k} \frac{|\phi_n(a)|^2 (E+\mu^2)}{(E_n-E)(E_n +\mu^2)}  + \frac{|\phi_{k}(a)|^2}{E_k + \mu^2} \right|<1$, the first term in the above equation becomes
\begin{eqnarray}
- \frac{\phi_k(x) \overline{\phi_k(y)}}{|\phi_k(a)|^2}  \left(\frac{1}{\alpha_R}-\sum_{n \neq k} \frac{|\phi_n(a)|^2 (E+\mu^2)}{(E_n-E)(E_n +\mu^2}) + \frac{|\phi_{k}(a)|^2}{E_k + \mu^2} 
\right) + O(|E_k-E|^2)
\end{eqnarray}
so that $G(x,y|E)$ is regular near $E=E_k$ as long as $\phi_k(a)\neq 0$. The uniqueness of the solution can be proved by showing that the sum is an increasing function of $E$ and goes to $-\infty$ as $E \to -\infty$, see Appendix C in \cite{Annals23} for the technical details.
\end{proof}

Similar results for a particular class of potentials have been examined in \cite{Grosche2} in the context of path integrals (in two and three dimensions). However, there is no explicit derivation  showing that the poles of the free resolvent are canceled in the final expression for the Green's function.

\begin{myremark}
Note that these results can be interpreted as a generalization of the well-known Sturm comparison theorems to the singular $\delta$ interactions, it is remarkable that even the renormalized case has this property. 
\end{myremark}

\begin{myremark} One would wonder how the separation between consecutive eigenvalues grow as we increase the index. There are some estimates if one knows how the manifold is isometrically embedded into an Euclidean space, see for example Theorem 5.6 in \cite{Urakawa-thebook}.
\end{myremark}



\section{Orthogonality Relation}

Using a contour integral of the resolvent $R(E)=(H-E)^{-1}$ around each simple eigenvalue $E_{k}^{*}$, we can find the projection operator onto the eigenspace associated with the eigenvalue $E_{k}^{*}$,
\begin{eqnarray}
    \mathbb{P}_k =  -\frac{1}{2\pi i} \oint_{\Gamma_k} R(E) \; d E \;,
\end{eqnarray}
where $\Gamma_k$ is the counter-clockwise oriented closed contour around each simple pole $E_{k}^{*}$, or equivalently 
\begin{eqnarray}
    \psi_k(x) \overline{\psi_k(y)} = -\frac{1}{2\pi i} \oint_{\Gamma_k} G(x,y|E) \; d E \;. \label{Projection}
\end{eqnarray}
From the explicit expression of the Green's function (\ref{fullGreenfunction}) and the residue theorem, we obtain 
\begin{eqnarray}
    \psi_k(x)=  \frac{G_0(x,a|E_{k}^{*})}{ \left(-\frac{d \Phi(E)}{d E}\bigg|_{E=E_{k}^{*}}\right)^{1/2}} \;. \label{boundstatewavefunctionexact}
\end{eqnarray}
Note that the differentiation under the summation yields
\begin{equation}
    {d\Phi(E)\over dE}\Big|_{E_k^*}= - \sum_{n=0}^\infty { |\phi_n(a)|^2\over (E_n-E_k^*)^2} \;. \label{derivativePhi}
\end{equation}
If $\phi_k(a)=0$, this term is skipped in the sum ensuring the expression being well-defined in all these cases. Moreover, in these special cases  then, the corresponding eigenfunction becomes,
\begin{equation}
    \psi_k(x)=\phi_k(x)
.\end{equation}
\begin{myproposition} Let $\phi_n$ be orthonormal set of eigenfunctions of $H_0$, i.e., 

\begin{eqnarray}
H_0 \phi_n &=&E_n \phi_n \nonumber \\ \int_{\mathcal{M}} \overline{\phi_n(x)} \phi_m (x) \; d\mu(x) & =&  \delta_{nm}.
\end{eqnarray} 
Then, the eigenfunctions $\psi_n$ of H, which is formally $H_0$ modified by a delta interaction supported at $x=a$ are orthonormal, that is, 
\begin{eqnarray}
    \int_{\mathcal{M}} \overline{\psi_n(x)} \psi_m (x) \; d\mu( x) = \delta_{nm} \;,
\end{eqnarray}
where $D=1,2,3$. \end{myproposition}

\begin{proof} We first prove for $D=2,3$, where the renormalization is needed to define point delta interactions properly. 

Using bilinear expansion (\ref{greenfuncexpansion}) of the Green's function of $H_0$ and the eigenfunction (\ref{boundstatewavefunctionexact}), we obtain
\begin{eqnarray}
    & &  \int_{\mathcal{M}} \overline{\psi_n(x)} \psi_m (x) \; d\mu( x) =  \int_{\mathcal{M}} \frac{\overline{G_0(x,a|E_{n}^{*})}}{ \left(-\frac{d\Phi(E)}{d E}\bigg|_{E=E_{n}^{*}}\right)^{1/2}} \frac{G_0(x,a|E_{m}^{*})}{\left(-\frac{d \Phi(E)}{d E}\bigg|_{E=E_{m}^{*}}\right)^{1/2}}  \; d\mu(x) \nonumber \\  & & = \frac{1}{\left(-\frac{d \Phi(E)}{d E}\bigg|_{E=E_{n}^{*}}\right)^{1/2} \left(-\frac{d \Phi(E)}{d E}\bigg|_{E=E_{m}^{*}}\right)^{1/2}}\int_{\mathcal{M}} \sum_k \frac{\phi_k(a) \overline{\phi_k(x)}}{E_k-E_{n}^*}  \sum_l \frac{\phi_l(x) \overline{\phi_l(a)}}{E_l-E_{m}^*}  \; d\mu(x) \;.
\end{eqnarray}
Interchanging the order of summation and integration and using the fact that $\phi_k$'s are orthonormal functions, we have
\begin{eqnarray}
 \int_{\mathcal{M}} \overline{\psi_n(x)} \psi_m (x) \; d\mu( x) = \frac{1}{\left(-\frac{d \Phi(E)}{d E}\bigg|_{E=E_{n}^{*}}\right)^{1/2} \left(-\frac{d \Phi(E)}{d E}\bigg|_{E=E_{m}^{*}}\right)^{1/2}} \sum_k \frac{|\phi_k(a)|^2}{(E_k-E_{n}^*)(E_k-E_{m}^*)} \;. \label{orth2}
\end{eqnarray}
If $n=m$ in \ref{orth2}, then it is easy to show that the new eigenfunctions $\psi_n$'s are automatically normalized thanks to the identity (\ref{derivativePhi}):
\begin{eqnarray}
    \int_{\mathcal{M}}  |\psi_n(x)|^2 d\mu(x) = -\frac{1}{\frac{d \Phi(E)}{d E}\bigg|_{E=E_{n}^{*}}} \sum_{k=0}^{\infty} 
    \frac{|\phi_k(a)|^2}{(E_k-E_n^{*})^2} = 1 \;.
\end{eqnarray}
For the case $n \neq m$, we first formally decompose the expression in the summation with a cut-off $N$ as a sum of two partial fractions
\begin{eqnarray}
    \sum_{k=0}^{N} \frac{|\phi_k(a)|^2}{(E_k-E_{n}^*)(E_k-E_{m}^*)} =  \sum_{k=0}^{N} \frac{|\phi_k(a)|^2}{(E_{n}^{*}-E_{m}^*)} \left( \frac{1}{E_k-E_{n}^*}-\frac{1}{E_k-E_{m}^*}\right) \;.
\end{eqnarray}
As explained in the renormalization procedure, each term $\sum_{k=0}^{N} \frac{|\phi_k(a)|^2}{E_k-E_{n}^*}$ is divergent as $N \to \infty$. Motivated by this, we add and subtract $\frac{1}{\alpha_R} + \sum_{k=0}^{N} \frac{|\phi_k(a)|^2}{E_k + \mu^2}$ to the above expression and obtain in the limit $N \to \infty$
\begin{eqnarray} \int_{\mathcal M} \overline{\psi_n(x)} \psi_m (x) \; d\mu( x) = \frac{1}{\left(E_{n}^* - E_{m}^* \right)} \frac{\left(\Phi(E_{n}^{*}) - \Phi(E_{m}^*)\right)}{\left(-\frac{d\Phi(E)}{d E}\bigg|_{E=E_{n}^{*}}\right)^{1/2} \left(-\frac{d \Phi(E)}{d E}\bigg|_{E=E_{m}^{*}}\right)^{1/2}} \;.
\end{eqnarray}
Since the zeroes of the function $\Phi$ are the bound state of the modified system, that is, $\Phi(E_{n}^{*})=0$ and $\Phi(E_{m}^{*})=0$ for all $n, m$ (when $n\neq m$), this completes our proof of the orthogonality of eigenfunctions for the modified  Hamiltonian having discrete spectrum.

The case for $D=1$ can easily be proved by following the same steps introduced above, except that there is no need for renormalization. 
\end{proof}

\begin{myremark}
If it so happens that for some $k$, $\phi_k(a)=0$, then the corresponding eigenvalue does not change, moreover the eigenfunction remains the same as $\phi_k(x)$. In this case, we see that the orthogonality among all the eigenfunctions continues to hold as well thanks to $\phi_k(a)=0$ again.
\end{myremark}

\section{Completeness Relation}

\begin{myproposition}
Let $\phi_n$ be a complete set of eigenfunctions of $H_0$, i.e., 
\begin{eqnarray}
H_0 \phi_n &=&E_n \phi_n \nonumber \\   \sum_{n=0}^{\infty} \overline{\phi_n(x)} \phi_n (y) &=& \delta(x-y)\;.
\end{eqnarray}
Then, the eigenfunctions $\psi_n$ of $H$, which is formally $H_0$ modified by a delta interaction supported at $x=a$, form a complete set, that is, 
\begin{eqnarray}
    \sum_{n=0}^{\infty} \overline{\psi_n(x)} \psi_n (y) = \delta(x-y)\;.
\end{eqnarray}
\end{myproposition}
\begin{proof} Let $\Gamma_n$ be the counter-clockwise oriented closed contours around each simple pole $E_{n}^{*}$ and $\Gamma_n \cap \Gamma_m = \emptyset$ for $n \neq m$, as shown in Figure \ref{fig:1}.

\begin{figure}[h!]
    \begin{center}
\begin{tikzpicture}[scale=0.8,
    decoration={
    markings,
    mark=between positions 1.2 and 1.8 step 0.2 with {\arrowreversed{stealth}}}
                ]

\draw[->, red, ultra thick] (-7, 0) arc [radius = 1cm, start angle= 0, end angle= 360];
\draw[->, red, ultra thick] (-4, 0) arc [radius = 1cm, start angle= 0, end angle= 360];
\draw[->, red, ultra thick] (-1, 0) arc [radius = 1cm, start angle= 0, end angle= 360];
\draw[->, red, ultra thick] (6, 0) arc [radius = 1cm, start angle= 0, end angle= 360];

        \draw[->] (-9.5,0) -- (9,0);
          \draw[->] (0,-3) -- (0,3);
		\node [scale=1] at (-8,0) {${\color{red}\times}$};
        \node [scale=1] at (-5,0) {${\color{red}\times}$};
        \node [scale=1] at (-2,0) {${\color{red}\times}$};
        \node [scale=1.5] at (1, 0.2) {${\color{red}\cdots}$};
        \node [scale=1] at (5,0) {${\color{red}\times}$};
        \node [scale=1.5] at (8, 0.2) {${\color{red}\cdots}$};
        \node  at (9,-0.5) {$\Real(E)$};
         \node  at (0.7,3) {$\Imaginary(E)$};
\node  at (-8,-0.5) {$E_{0}^*$};
\node  at (-5,-0.5) {$E_{1}^*$};
\node  at (-2,-0.5) {$E_{2}^*$};
\node  at (5,-0.5) {$E_{n}^*$};

\node at (3.5,0) [circle,fill,inner sep=1.5pt]{};
\node  at (3.5,-0.5) {$E_{n-1}$};

\node at (-6.5,0) [circle,fill,inner sep=1.5pt]{};
\node  at (-6.5,-0.5) {$E_{0}$};

\node at (-3.5,0) [circle,fill,inner sep=1.5pt]{};
\node  at (-3.5,-0.5) {$E_{1}$};

\node at (-0.5,0) [circle,fill,inner sep=1.5pt]{};
\node  at (-0.5,-0.5) {$E_{2}$};

    \end{tikzpicture}     
\end{center}
\caption{The contours $\Gamma_n$ along each simple pole $E_{n}^{*}$ with counterclockwise orientation.}
    \label{fig:1}
\end{figure}

Then, the projection onto the associated eigenspace is given by the formula (\ref{Projection}), and thanks to Krein's formula for the Green's function of the modified Hamiltonian (\ref{fullGreenfunction}), we have
\begin{eqnarray}
   \sum_{n=0}^{\infty} \overline{\psi_n(x)} \psi_n (y) =   \frac{1}{2\pi i} \sum_{n=0}^{\infty} \oint_{\Gamma_n \supset E_{n}^*} \left(G_0(x,y|E)+\frac{G_0(x,a|E) G_0(a,y|E)}{\Phi(E)}  \right) \, dE \;.
\end{eqnarray}
Note that the total expression in the Krein's formula has only poles at $E_n^*$'s, when we think of it as the sum of two separate expressions, we have the original eigenvalues, $E_n$, reappearing as poles again. 
Here the contribution coming from the Green's function of the initial Hamiltonian $H_0$, which is the first term of Krein's formula,  for the above contour integral  {\it vanishes since the poles $E_n$ of $G_0$ are all located outside} at each $\Gamma_n$ (note that in the special case of coincidence of one $E_k^*$ with $E_k$, $\phi_k(a)=0$, so that the contribution of the other term is zero and we pick the original wavefunctions $\phi_k(x)$, so in such cases we exclude these terms from the summation and write them separately). For simplicity, we assume that all $E_k^*\neq E_k$ from now on.
Note that thanks to the denominators we can elongate the contours to ellipses that extend to infinity along the imaginary direction (on the complex $E$-plane). We now continuously deform this contour to the following extended contour $\Gamma_{snake}$, as shown in Figure \ref{fig:2}. Note that we have {\it no poles of the Green's function on the left part of the line $E_{0}^{*} + i {\mathbb R}$ nor any zeros of} $\Phi(E)$, the product of two Green's functions decay rapidly as $|E|\to \infty$ along the negative real direction as well as along the imaginary directions, hence we have no contributions from the contours at infinity for these deformations. This observation allows us to change the contour as described below.

Using the interlacing theorem stated in  Proposition \ref{prop21}, we can, so to speak,  flip the contour while preserving the value of the integration and then deform the contour to the one $\Gamma_{dual}$ that consists of isolated closed contours $\Gamma_{dual}^{n}$ around each isolated eigenvalue $E_n$ of the initial Hamiltonian $H_0$ 
with opposite orientation, as shown in Figure \ref{fig:5}.
\begin{figure}[h!]
    \begin{center}
\begin{tikzpicture}[scale=0.8,
    decoration={
    markings,
    mark=between positions 1.2 and 1.8 step 0.2 with {\arrowreversed{stealth}}}
                ]

\draw[<-, red, ultra thick] (-7, 0) arc [radius = 1cm, start angle= 0, end angle= -180];
\draw[<-, red, ultra thick] (-9, 0) -- (-9, 1);

\draw[->, red, ultra thick] (-7, 0) arc [radius = 0.5cm, start angle= 190, end angle= -10];

\draw[<-, red, ultra thick] (-4, 0) arc [radius = 1cm, start angle= 0, end angle= -180];
\draw[->, red, ultra thick] (-4, 0) arc [radius = 0.5cm, start angle= 190, end angle= -10];

\draw[<-, red, ultra thick] (-1, 0) arc [radius = 1cm, start angle= 0, end angle= -180];

\draw[<-, red, ultra thick] (5.9, 0) arc [radius = 1cm, start angle= 0, end angle= -180];

\draw[->, red, ultra thick] (2.9, 0) arc [radius = 0.5cm, start angle= 190, end angle= -10];

        \draw[->] (-9.5,0) -- (9,0);
          \draw[->] (0,-3) -- (0,3);
		\node [scale=1] at (-8,0) {${\color{red}\times}$};
        \node [scale=1] at (-5,0) {${\color{red}\times}$};
        \node [scale=1] at (-2,0) {${\color{red}\times}$};
        \node [scale=1.5] at (1, 0.2) {${\color{red}\cdots}$};
        \node [scale=1] at (5,0) {${\color{red}\times}$};
        \node [scale=1.5] at (8, 0.2) {${\color{red}\cdots}$};
        \node  at (9,-0.5) {$\Real(E)$};
         \node  at (0.7,3) {$\Imaginary(E)$};
\node  at (-8,-0.5) {$E_{0}^*$};
\node  at (-5,-0.5) {$E_{1}^*$};
\node  at (-2,-0.5) {$E_{2}^*$};
\node  at (5,-0.5) {$E_{n}^*$};

\node at (3.5,0) [circle,fill,inner sep=1.5pt]{};
\node  at (3.5,-0.5) {$E_{n-1}$};

\node at (-6.5,0) [circle,fill,inner sep=1.5pt]{};
\node  at (-6.5,-0.5) {$E_{0}$};

\node at (-3.5,0) [circle,fill,inner sep=1.5pt]{};
\node  at (-3.5,-0.5) {$E_{1}$};

\node at (-0.5,0) [circle,fill,inner sep=1.5pt]{};
\node  at (-0.5,-0.5) {$E_{2}$};


    \end{tikzpicture}     
\end{center}
\caption{The contour $\Gamma_{snake}$}
    \label{fig:2}
\end{figure}

\begin{figure}[h!]
    \begin{center}
\begin{tikzpicture}[scale=0.8,
    decoration={
    markings,
    mark=between positions 1.2 and 1.8 step 0.2 with {\arrowreversed{stealth}}}
                ]       
  
\draw[<-, red, ultra thick] (-5.5, 0) arc [radius = 1cm, start angle= 0, end angle= 360];
\draw[<-, red, ultra thick] (-2.5, 0) arc [radius = 1cm, start angle= 0, end angle= 360];
\draw[<-, red, ultra thick] (0.5, 0) arc [radius = 1cm, start angle= 0, end angle= 360];
\draw[<-, red, ultra thick] (4.5, 0) arc [radius = 1cm, start angle= 0, end angle= 360];

        \draw[->] (-9.5,0) -- (9,0);
          \draw[->] (0,-3) -- (0,3);
		\node [scale=1] at (-8,0) {${\color{red}\times}$};
        \node [scale=1] at (-5,0) {${\color{red}\times}$};
        \node [scale=1] at (-2,0) {${\color{red}\times}$};
        \node [scale=1.5] at (1.5, 0.2) {${\color{red}\cdots}$};
        \node [scale=1] at (5,0) {${\color{red}\times}$};
        \node [scale=1.5] at (7, 0.2) {${\color{red}\cdots}$};
        \node  at (9,-0.5) {$\Real(E)$};
         \node  at (0.7,3) {$\Imaginary(E)$};
\node  at (-8,-0.5) {$E_{0}^*$};
\node  at (-5,-0.5) {$E_{1}^*$};
\node  at (-2,-0.5) {$E_{2}^*$};
\node  at (5,-0.5) {$E_{n}^*$};

\node at (3.5,0) [circle,fill,inner sep=1.5pt]{};
\node  at (3.5,-0.5) {$E_{n-1}$};

\node at (-6.5,0) [circle,fill,inner sep=1.5pt]{};
\node  at (-6.5,-0.5) {$E_{0}$};

\node at (-3.5,0) [circle,fill,inner sep=1.5pt]{};
\node  at (-3.5,-0.5) {$E_{1}$};

\node at (-0.5,0) [circle,fill,inner sep=1.5pt]{};
\node  at (-0.5,-0.5) {$E_{2}$};

    \end{tikzpicture}     
\end{center}
\caption{The contours $\Gamma_{dual}^{n}$ along each simple pole $E_{n}$ with clockwise orientation.}
    \label{fig:5}
\end{figure}

Hence, we have
\begin{eqnarray}
   \sum_{n=0}^{\infty} \overline{\psi_n(x)} \psi_n (y) =   \frac{1}{2\pi i} \sum_{n=0}^{\infty} \oint_{\Gamma_{dual}^{n} \supset E_{n}} \frac{G_0(x,a|E) G_0(a,y|E)}{\Phi(E)}  \, dE \;.
\end{eqnarray}
We then assume that all isolated closed contours $\Gamma_{dual}^{n}$ are sufficiently small. To be more precise, one must consider the truncated sum, for the sake of clarity we ignore this subtlety for now. Then, the above expression can be written as
\begin{eqnarray}
     \frac{1}{2\pi i} \sum_{n=0}^{\infty} \oint_{\Gamma_{dual}^{n} \supset E_{n}} \frac{G_0(x,a|E) G_0(a,y|E)}{\frac{1}{\alpha_R} + \sum_{l=0}^{\infty} \frac{|\phi_l(a)|^2}{E_l+\mu^2} -\frac{|\phi_n(a)|^2}{E_n -E}-\sum_{l \neq n}^{\infty} \frac{|\phi_l(a)|^2}{E_l-E} }\, dE \;.
\end{eqnarray}
As we know from the proof of cancellation of poles (in our previous work), we split the above expression in the following way 
\begin{eqnarray*}
   & &   \frac{1}{2\pi i} \sum_{n=0}^{\infty} \oint_{\Gamma_{dual}^{n} \supset E_{n}} \Big( g_n(x,a|E) + \frac{\overline{\phi_n(a)} \phi_n(x)}{E_n-E}\Big) \\
  & & \hspace{4cm} \times   \Big( \frac{(E_n-E)}{D_n(\alpha_R,E)(E_n-E)-|\phi_n(a)|^2}\Big) \Big(g_n(a,y|E) + \frac{\overline{\phi_n(y)} \phi_n(a)}{E_n-E}\Big) \, dE \;,
\end{eqnarray*}
where the functions $g_n$ and $D_n$ are regular/holomorphic inside for each one of $\Gamma_{dual}^n$, 
which are defined near $E=E_n$ for a given  $n$ as:
\begin{eqnarray}
    g_n(x,y|E) & := & \sum_{ k\neq n} \frac{\phi_k(x) \overline{\phi_k(y)}}{E_k-E} \;, \\
    D_n(\alpha, E) & := & \frac{1}{\alpha}-\sum_{k \neq n} \frac{|\phi_k(a)|^2}{E_k-E}  \;.
\end{eqnarray}
Then, the above integral must have the following form:
\begin{eqnarray*}
   & &   \frac{1}{2\pi i} \sum_{n=0}^{\infty} \oint_{\Gamma_{dual}^{n} \supset E_{n}} \Big( \text{holomorphic \, part} + \frac{|\phi_n(a)|^2 \overline{\phi_n(y)} \phi_n(x)}{E_n-E}\Big)  \Big( \frac{1}{D(\alpha_R,E)(E_n-E)-|\phi_n(a)|^2}\Big) \, dE \;.
\end{eqnarray*}
Applying the residue theorem, we obtain 
\begin{eqnarray}
   \sum_{n=0}^{\infty} \overline{\psi_n(x)} \psi_n (y) =   \frac{1}{2\pi i} \sum_{n=0}^{\infty} \frac{\phi_n(x) \overline{\phi_n(y)}}{-|\phi_n(a)|^2} \left( -2\pi i |\phi_n(a)|^2 \right) \;,
\end{eqnarray}
where the minus sign is due to the opposite orientation of the contour $\Gamma_{dual}$. Finally (which should be done in a more rigorous way by  taking a limit of truncated expressions), we prove
\begin{eqnarray}
   \sum_{n=0}^{\infty} \overline{\psi_n(x)} \psi_n (y) =      \sum_{n=0}^{\infty} \overline{\phi_n(x)} \phi_n (y) = \delta(x-y) \;. 
\end{eqnarray}
\end{proof}

\begin{myremark} As explained above, when for a particular value $k$, $\phi_k(a)=0$, our proof can be modified, by separating this eigenfunction in the Green function and then deforming the contours accordingly. In our previous work \cite{Annals23},  possible degeneracy (corresponding to a $d$ dimensional eigensubspace) is also discussed for a singular interaction. When all the degenerate eigenvectors are zero at $a$, there is no effect of the singular interaction; hence we can separate this projection and repeat our proof.  If that is not the case, then the singular interaction lifts the degeneracy in a particular direction, as explained precisely in \cite{Annals23}. The eigenvector in this particular direction changes to $G_0(x,a|E^*)$, where $E^*$ refers to the new eigenvalue appearing in the spectrum,  and the other orthogonal directions, forming a $d-1$ dimensional subspace, are left intact.  So our proof goes through in this case as well by separating the unaffected projection and repeating our proof accordingly.  \end{myremark}

\begin{myremark}
Interestingly, these observations lead to an explicit construction of the resulting renormalized Hamiltonian. Suppose that there is a set of $\phi_k(x)$ for which we have $\phi_k(a)=0$, call this set of indices as ${\mathcal N}$, nodal indices, then  the renormalized Hamiltonian becomes (as an integral operator)
\begin{equation}
\langle x|H| y\rangle=\sum_{ k\notin {\mathcal N}}^\infty E_k^* \left( {d\Phi(E)\over dE}\Big|_{E_k^*}\right)^{-1}G_0(x, a|E_k^*)G_0(a, y|E_k^*)
+\sum_{k \in \mathcal{N}}E_k\overline{ \phi_k(x)}\phi_k(y) \;. \end{equation}    
\end{myremark}
\begin{myremark}
Incidentally, the above integral kernel can be utilized to show that the operator $H$, defined through this kernel,  is {\it essentially self-adjoint}  thanks to the example 9.25  given in \cite{Hall} and stated (somewhat more intuitively) below for convenience.  
\end{myremark}
Suppose we have a symmetric (what physicists typically call Hermitian) operator $A$ which has a complete set of eigenvectors, then the closure of operator $A$, that is if we define $A$ on a slightly larger set, by adding all vectors for which 
$A$ acts continuously to its domain, becomes a self-adjoint operator, see e.g., \cite{Michelangeli} for a pedagogical discussion of this.
Note that the above expression does not manifest $H$ as a {\it perturbation or modification} of $H_0$, it may be possible to reexpress this kernel as $\langle x|H_0|y\rangle+\delta_R(x,y)$, for some function $\delta_R$ which is not in the domain of $H_0$.
Alternatively, we can rewrite the Hamiltonian as an abstract operator,
\begin{equation}
H=\sum_{ k\notin {\mathcal N}}^\infty E_k^*\, (H_0-E_k^*)^{-1}| a\rangle \left( {d\Phi(E)\over dE}\Big|_{E_k^*}\right)^{-1}\langle a| (H_0-E_k^*)^{-1}+
    \sum_{k \in \mathcal{N}}E_k|\phi_k\rangle\langle \phi_k|
.\end{equation}
It is clear that the resulting (renormalized) operator cannot be expressed as a differential operator, but only as an integral operator.

\begin{myremark}
    Using the development in our previous work \cite{Annals23}, the present discussion can be easily extended to $N$ center case, the case where delta interaction is supported on curves in the plane or space etc. In principle, all these extensions are possible and left as an exercise for an enthusiastic reader to get involved with singular interactions.    
\end{myremark}

\begin{myproposition} The set of functions $G_0(x, a|E_k^*)-G_0(x,a|E_l^*)$ are in the domain of the initial Hamiltonian $H_0$.
\end{myproposition}

\begin{proof}
The difference in the Green's functions can be written explicitly as follows,
\begin{eqnarray}
   \xi(x)=G_0(x, a|E_k^*)-G_0(x,a|E_l^*)=(E_k^*-E_l^*)\sum_{n=0}^\infty {\phi_n(x)\phi_n(a)\over (E_n-E_k^*)(E_n-E_l^*)} \;.
\end{eqnarray}
Suppose $E_k^*>E_l^*$ and since $E_n\to \infty$ as $n\to \infty$, monotonously, we choose $N_*$ such that $E_n> 3E_k^*$ for $n\geq N_*$. This implies that 
$E_n-E_k^*> {1\over 2} (E_n+E_k^*)$. 
Let us compute formally
$ ||H_0 \xi||^2$:
\begin{equation}
    \int_{\mathcal M} d\mu(x) | (H_0\xi)(x)|^2= (E_k^*-E_l^*)^2 \sum_{n=0}^\infty { E_n^2 |\phi_n(a)|^2\over (E_n-E_k^*)^2(E_n-E_l)^2}
\; .\end{equation}
We split the sum into two parts
\begin{eqnarray}
||H_0\xi||^2&=& (E_k^*-E_l^*)^2 \left(\sum_{n=0}^{N_*} {E_n^2|\phi_n(a)|^2\over (E_n-E_k^*)^2(E_n-E_l^*)^2}+\sum_{n={N_*}}^\infty {E_n^2 |\phi(a)|^2\over (E_n-E_k^*)^2(E_n-E_l^*)^2} \right) \nonumber\\
        &<& (E_k^*-E_l^*)^2 \left( \sum_{n=0}^{N_*} {E_n^2|\phi_n(a)|^2\over (E_n-E_k^*)^2(E_n-E_l^*)^2}+\sum_{n={N_*}}^\infty {E_n^2 |\phi_n(a)|^2\over (E_n-E_k^*)^4} \right) \nonumber \\
        &<& (E_k^*-E_l^*)^2 \left( \sum_{n=0}^{N_*} {E_n^2|\phi_n(a)|^2\over (E_n-E_k^*)^2(E_n-E_l^*)^2}+2\sum_{n={N_*}}^\infty {E_n^2 |\phi_n(a)|^2\over (E_n+E_k^*)^4} \right) \;.      \label{H0norm}
\end{eqnarray}
Use now $E_n^2=(E_n+E_k^*)^2-2(E_n+E_k^*) E_{k}^{*}+(E_k^*)^2$, to reexpress the last part as
\begin{eqnarray}
    \sum_{n={N_*}}^\infty {E_n^2 |\phi_n(a)|^2\over (E_n+E_k^*)^4}&=&\sum_{n={N_*}}^\infty { |\phi_n(a)|^2\over (E_n+E_k^*)^2}-2 E_{k}^{*}\sum_{n={N_*}}^\infty { |\phi_n(a)|^2\over (E_n+E_k^*)^3}
+(E_k^*)^2\sum_{n={N_*}}^\infty { |\phi_n(a)|^2\over (E_n+E_k^*)^4} \;.
\end{eqnarray}
Removing the negative term (as all its summands are positive it gives an upper bound to our expression) and adding the missing terms in the sums so as to turn them into the sum over from $n=0$ to $n=\infty$, we find an upper bound for the last term in (\ref{H0norm}): 
\begin{eqnarray}
    \sum_{n={N_*}}^\infty {E_n^2 |\phi_n(a)|^2\over (E_n+E_k^*)^4}&<&\sum_{n=0}^\infty { |\phi_n(a)|^2\over (E_n+E_k^*)^2}+(E_k^*)^2\sum_{0}^\infty { |\phi_n(a)|^2\over (E_n+E_k^*)^4}\nonumber\\
    &<& \int_0^\infty t \; K_t(a,a) e^{-{E_k^*}t} \; d t +{E_k^*}^2\int^\infty_0 t^3 \; K_t(a,a) e^{-{E_k^*}t} \; d t \;,
\end{eqnarray}
where we have used $\frac{1}{(E_n+E_{k}^*)^k} = \int_{0}^{\infty} t^{k-1} e^{-t(E_n + E_{k}^{*})} \; dt$ and the eigenfunction expansion of the heat kernel $K_t(x,y)= \sum_{n=0}^{\infty} \phi_n(x) \overline{\phi_n(y)} e^{-tE_{n}}$. Using the upper bound for the diagonal heat kernel on compact Riemannian manifolds $K_t(a,a) \leq \frac{1}{V(\mathcal{M})} + C t^{-D/2}$, where $V(\mathcal{M})$ is the volume of the manifold and $C$ is a positive constant depending on the geometry of the manifold such as the bounds on Ricci curvature given by (\ref{ricci}), it is easy to see that all the integrals above are finite. The same bound has been also used for showing the lower bound for the ground state energy of a particle interacting with finitely many delta interactions on compact manifold \cite{ErmanTurgut}. Moreover, since the first term of the sum being over a finite number of indices in (\ref{H0norm}) is finite, we show that $||H_0 \xi ||$ is finite. In other words, $\xi$ is in the domain of $H_0$.  
\end{proof}

\begin{myremark}
The explicit realization above gives us some insight about the self-adjoint extension perspective as well. Note that the 
functions $G_0(x, a|E_k^*)$'s are not in the domain of the initial Hamiltonian $H_0$, nevertheless we have shown that their difference $G_0(x, a|E_k^*)-G_0(x,a|E_l^*)$ are in the domain 
of $H_0$, hence we need only one of them to be added to the initial domain $D(H_0)$.
\end{myremark}

\begin{myremark}
It is possible to provide the upper and lower bounds for  these new eigenfunctions on manifolds, which charaterize the singular behavior as $x\to a$. Considering manifolds with Ricci bounded from below by the metric,  for $d=3$ we have,
$$
-C_0+{C_1\over d_g(x,a)}\leq G_0(x,a|E_k^*)\leq {C_2\over d_g(x,a)}.
$$
When $d=2$, for compact manifolds Ricci boundedness is automatically true, and we get a logarithmic bound,
$$
-C_0+C_1\ln(d_g(x,a))\leq G_0(x,a|E_k^*)\leq C_0+C_2\ln(d_g(x,a)).
$$
For both estimates, the constants $C_0, C_1, C_2$ depend only on the  dimension and geometric data such as the volume, diameter and the value of the lower bound constant on the Ricci curvature (however in a physical problem there are also $\hbar^2$ and $m$ multiplicative factors appearing in these bounds). \end{myremark}

\section{Application: Sudden Approximation in the Case of Time-dependent Center}

    We note that the above explicit expression for the wave functions can be used for an interesting application; suppose that we initially have our delta-modification at point $a$ and very rapidly we move this modification to another point $b$. We can use the usual sudden perturbation approach to this problem just as in the conventional case.

We briefly elaborate on this idea, let us suppose that initially the system is prepared in the eigenstate $G_0(x, a|E_k^*(a))$,  $E_k^*(a)$ referring to the energy for this case. A sudden perturbation means that the system has no time to readjust itself, so the wave function remains as it is, but should be decomposed in terms of the new eigenbasis $G_0(x,b|E_m^*(b))$'s to calculate the probability of finding the system in the new energy eigenstate $E_m^*(b)$. This means that the conditional probability of finding the system in $E_m^*(b)$, given that it was in $E_k^*(a)$ initially, is
\begin{eqnarray}
    p(m,b|k,a)&=&\left[ {d\Phi(E|a)\over dE}\Big|_{E_k^*} {d\Phi(E|b)\over dE}\Big|_{E_m^*}\right]^{-1}\left|\int_{\mathcal M} d\mu(x) \overline{G_0(x, b|E_m^*(b))}G_0(x, a|E_k^*(a))\right|^2 \nonumber\\
    &=& \left[ {d\Phi(E|a)\over dE}\Big|_{E_k^*} {d\Phi(E|b)\over dE}\Big|_{E_m^*}\right]^{-1}\left | {G_0(a,b|E_m^*(b))- G_0(a,b|E_k^*(a))\over E_m^*(b)-E_k^*(a)}\right|^2\, ,\nonumber
    \end{eqnarray}
where the energy eigenstates $E_m^*(b)$
are found from the solutions of 
\begin{equation}
    \Phi(E|b)=\frac{1}{\alpha_R} -\sum_k{ |\phi_k(b)|^2(E+\mu^2)\over (E_k+\mu^2)(E_k-E)}=0,
\end{equation}
whereas $E_k^*(a)$ refers to the zeros of $\Phi(E|a)$. Incidentally, it is possible to conceive a sudden change of $a$ and $\mu_a$ to $b$ and $\mu_b$, without any difficulty. As pointed out before, one can easily generalize this idea to sudden changes of curves in three dimensions, or sudden rearrangements of multiple centers etc. The sudden approximation is typically valid if the time scale, defined by the initial energy eigenstate $E_k^*(a)$ is much larger than the time scale of the change we consider.


\begin{myremark}
The above results are independent of the chosen renormalization scheme, as shown in \cite{DuttaRoy} for the point delta interactions in two and three dimensions. The main idea of the proof for the completeness of the eigenfunctions of the Hamiltonian involving singular delta potentials here is based on the eigenfunction expansion of Green's function $G_0$ and the contour deformation described above.    
\end{myremark}

\section{Acknowledgements}
We would like to thank A. Michelangeli, M. Znojil for their interest in our works and their continual support. We also thank A. Mostafazadeh, M. Gadella, K. G. Akbas, E. Ertugrul and S. Seymen for discussions. O. T. Turgut is grateful to M. Deserno for a wonderful time in Carnegie-Mellon University where this work began.
Last, but not least we thank 
P. Kurasov for all the inspiration that led to this work.

\end{document}